\documentclass[runningheads]{llncs}
\usepackage{algorithm}
\usepackage{algorithmic}
\usepackage{booktabs}
\usepackage{hyperref}
\usepackage{textcomp}
\usepackage{todonotes}
\usepackage{forest}
\usepackage{xspace}
\usepackage{bbding}

\usepackage[T1]{fontenc}

\def\orcidID#1{\href{http://orcid.org/#1}{\raisebox{-1.25pt}{\includegraphics{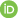}}}}

\newcommand{\leancop}{\textsf{leanCoP}\xspace}
\newcommand{\meancop}{\textsf{meanCoP}\xspace}
\newcommand{\meancopcut}{\textsf{!meanCoP}\xspace}
\newcommand{\hopcop}{\textsf{hopCoP}\xspace}

\usepackage{color}

\begin{document}
\title{\texorpdfstring{Constraint Learning for\\Non-Confluent Proof Search}{Constraint Learning for Non-Confluent Proof Search}}
\author{
Michael~Rawson\inst{1}\textsuperscript{(\Envelope)}\orcidID{0000-0001-7834-1567} \and
Clemens~Eisenhofer\inst{2}\orcidID{0000-0003-0339-1580} \and
Laura~Kov\'acs\inst{2}\orcidID{0000-0002-8299-2714}
}
\authorrunning{Rawson et al.}
\institute{
University of Southampton, Southampton, UK\\
\email{michael@rawsons.uk}
\and
TU Wien, Vienna, Austria\\
\email{\{clemens.eisenhofer,laura.kovacs\}@tuwien.ac.at}
}
\maketitle
\begin{abstract}
Proof search in non-confluent tableau calculi, such as the connection tableau calculus, suffers from excess backtracking, but simple restrictions on backtracking are incomplete. We adopt \emph{constraint learning} to reduce backtracking in the classical first-order connection calculus, while retaining completeness. An initial constraint learning language for connection-driven search is iteratively refined to greatly reduce backtracking in practice. The approach may be useful for proof search in other non-confluent tableau calculi.
\end{abstract}

\section{Introduction} State-of-the-art methods for automated theorem proving are based on exhaustive search, using a \emph{proof calculus} to explore the space of possible proofs.
The search for proofs can be either backtracking or non-backtracking in nature.
Backtracking search is required when the underlying calculus allows search to become ``stuck'' because of choices made previously in the search.
These previous choices must be undone, and an alternative choice made, in order for the search to continue, which we call \emph{backtracking}.

Some calculi do not require backtracking, such as confluent tableau calculi.
Calculi like superposition and instance generation also fall into this category.
Backtrack-free calculi are sometimes preferred and often enjoy theoretical advantages.
However, in some cases, a non-confluent calculus is more practically effective, or is preferred for some other reason.
We are therefore interested in \emph{improving the behaviour of backtracking proof search}.

Backtracking search is not a problem inherently: state-of-the-art SAT solvers are uniformly based on backtracking procedures.
Problems arise when the backtracking behaviour is pathological, backtracking too little, and trying to close the same goals again when the root cause of the dead end has not changed.
We note in passing that backtracking \emph{too much} would also cause problems.
In SMT solving, for example, adding or removing theory literals to or from their respective decision procedures is a relatively expensive operation that should be avoided if possible, which is why smaller backtracking steps are preferred.

We address this here by adapting a technique called \emph{constraint learning}\footnote{~better known as \emph{clause learning} in the context of Boolean satisfiability} from the constraint satisfaction community.
During the search for a closed tableau, we sometimes arrive at a dead end where no further inferences are applicable in the tableau calculus.
At this point, we analyse the \emph{reason} that no inference is applicable and \emph{learn} a constraint clause that prevents us from arriving at a similar tableau that is stuck for the same reason.
The accumulating constraint database helps to guide the search towards more promising areas, or eventually shows that no closed tableau exists.
\\
\\
\noindent
\fbox{
\begin{minipage}{\textwidth}
\textbf{A potential source of confusion.}
Readers familiar with SAT solving, instance generation and/or refutational theorem proving may suspect that we are learning consequences of the input problem, and that if we derive an obviously unsatisfiable constraint like the empty clause or $0 = 1$, the input problem is unsatisfiable.
This is \emph{not} the case: we are learning constraints about the search space, and such constraints show that there is no closed tableau to be found (at a particular resource bound).
\end{minipage}
}

\section{Background and Motivation}
\begin{figure}[t]
    \centering
    \begin{forest}
        [,
            [\(L_1\),]
                [\(L_2\),]
                [\dots,]
                [\(L_n\),]
        ]
    \end{forest}
    \hspace{.4in}
    \begin{forest}
        [\dots
        [\(L\), name=mate
        [\dots
        [\(L_1\), name=goal]
        [\dots]
        ]
        [\dots]
        ]
        [\dots]
        ]
        \draw[dashed] (goal) to[out=north west, in=west] (mate);
    \end{forest}
    \hspace{.4in}
    \begin{forest}
        [\dots
        [\(L\), name=goal
        [\(L_1\), name=mate]
        [\(L_2\),]
        [\dots,]
        [\(L_n\),]
        ]
        [\dots]
        ]
        \draw[dashed] (mate) to[out=north west, in=west] (goal);
    \end{forest}
    \caption{The three inference rules of the clausal connection tableau calculus: \emph{start}, \emph{reduction}, and \emph{extension}. In the \emph{start} and \emph{extension} rules, \(C = L_1 \vee L_2 \vee \dots \vee L_n\) is a clause from the input set, with its variables renamed apart from the tableau. In the \emph{reduction} and \emph{extension} rules we require that \(\sigma(\lnot{L}) = \sigma(L_1)\), i.e. $L$ and $L_1$ are \emph{connected} (shown with dashed lines).}
    \label{fig:rules}
\end{figure}
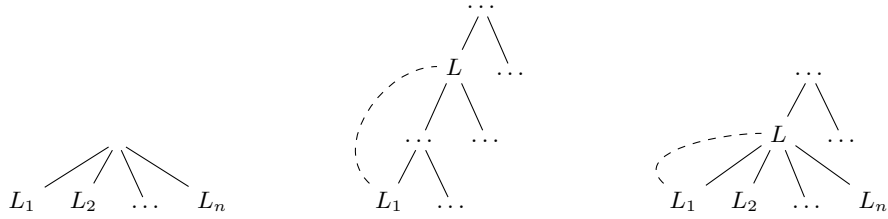

We assume familiarity with tableau methods~\cite{handbook-of-tableau-methods} and classical first-order logic.
The connection method~\cite{bibel}, connection tableau calculus~\cite{handbook}, model elimination~\cite{model-elimination}, and the method of matings~\cite{matings} are closely related proof search methods.
We will present our work in the language of connection \emph{tableaux}, given the audience.
The connection refinement demands that any addition to an open branch is \emph{connected} (contains a literal of opposite polarity) to the leaf literal, which produces a very strong goal-directed effect.
This comes at the expense of confluence, even for propositional logic: choosing the wrong extension of the tableau can prevent closing the tableau.

By lifting propositional connection tableau to a free-variable tableau calculus with a global substitution $\sigma$, one obtains a complete calculus for classical first-order logic~\cite{handbook}.
We show the three inference rules of the clausal connection tableau calculus in Figure~\ref{fig:rules}: \emph{start} adds a clause to an empty tableau, \emph{reduction} closes a branch by connecting a leaf literal to a literal on its path, and \emph{extension} adds a clause that is connected to an open branch.
With minor modifications, the connection method can be adapted to other logics such as intuitionistic~\cite{leancop} or modal~\cite{mleancop} logics, or to non-clausal proof search~\cite{nanocop}.

The simplicity of the calculus admits very compact theorem provers, often making use of Prolog and related technology~\cite{leancop,cop-survey,pttp}.
While tableau methods are no longer the state of the art in classical first-order theorem proving, they are still competitive for proving conjectures in the presence of large numbers of irrelevant axioms (a key application for interactive theorem provers~\cite{sledgehammer}) or in specialised settings~\cite{sgcd}.
They have also found a new home in experiments applying machine learning to theorem proving~\cite{rlcop,markov,plcop}, where their simplicity --- and to some extent their backtracking --- makes them an attractive choice.
We assume here that equality has been preprocessed away from the input~\cite{handbook-ar-paramodulation,brand,lazy-paramodulation-practice}, although it is possible to extend connection calculi with support for equational reasoning~\cite{bibel,lazy-paramodulation,breu}.

\subsection{Excess Backtracking in Connection-Driven Search}
In order to remain complete, propositional connection systems must consider alternative additions to the tableau, but once a branch has been closed, it can remain so.
At the first-order level, \emph{alternative ways to close the same branch} must also be considered: this is because closing a branch may bind variables in the global substitution in a way that prevents closing a different branch later.

Otten noticed that this requirement produced an enormous amount of backtracking in some cases~\cite{restricted-backtracking}.
He introduced a Prolog \emph{cut} into \leancop's search routine, rendering it incomplete but significantly reducing backtracking and increasing performance in many cases.
Later, Färber studied the behaviour of these cuts extensively and developed several variants~\cite{meancop}.
All of Färber's variants are incomplete, but some are considerably more effective in practice than others.
We will use his \meancop system as a point of comparison in Section~\ref{sec:implementation}.

\subsection{Terminology and Convention}
We use some terms informally, which we hope will aid understanding rather than cause confusion.
As tableau-based first-order theorem provers typically manipulate a tableau and a substitution together, we will refer to both of them simply as ``the \emph{tableau}'' where it is not confusing.
When an inference of the calculus is attempted but cannot be successfully applied, we say it has \emph{failed}.
Moreover, if we can detect that a tableau can never be closed, we say it is \emph{stuck}.
Both failed inferences and stuck tableaux may be \emph{explained} in terms of a constraint, which we call a \emph{reason}.

We will need to refer explicitly to particular positions in a tableau. We use the obvious scheme where for any position $p$, $p.i$ is the position at the $i$\textsuperscript{th} branch below $p$. The empty position stands for the root of the tableau.
First-order variables in a tableau are named according to the \emph{position} below which their clause is attached, and are hence \emph{consistent} across backtracking.
We simply use $u$, $v$, $w$, $x$, $y$ and $z$ for first-order variables, $c$ and $d$ for constants, $f$ for a function, and $P$, $Q$, $R$, and $S$ for predicates.
Finally, we elide parentheses in terms and literals and consider $\lnot \lnot L$ identical to literal $L$.

\subsection{Constraint Learning}
Constraint learning~\cite{constraint-learning} is a well-known but somewhat vaguely defined approach in constraint satisfaction and artificial intelligence.
It is not necessary for our purposes to formally define constraint learning nor explore all of its developments, but the core of the idea is as follows.
In a backtracking search for a solution to a set of constraints, we may encounter a dead end, where making a step in any available direction violates some constraint.
A subset of the search's previous decisions may be blamed for this situation by a justification extraction process. A constraint enforcing that not all of the elements within the justification may be selected simultaneously is added to the constraint set, which prevents search from running into a similar unfortunate situation again.
This \emph{learned} constraint (usually in the form of conflict clauses) can also be used to do \emph{backjumping}: backtracking by more than one level.

Constraint learning was particularly effective for Boolean satisfiability~\cite{grasp} (SAT) solving and is the basis for modern \emph{conflict-driven clause learning} (CDCL) SAT solvers and therefore satisfiability modulo theory (SMT) solvers~\cite{SMT}.

\subsection{Running Example}
We use a particular example throughout the paper. Consider the following set of first-order clauses:
\begin{equation}
\label{eq:c1}
\forall xyz.~Px \vee Qy \vee Rxy \vee Pz
\end{equation}
\begin{minipage}{.48\linewidth}
\begin{equation}
\label{eq:c2}
\forall x.~\lnot Px \vee S
\end{equation}
\begin{equation}
\label{eq:c3}
\lnot S \vee \lnot Pc
\end{equation}
\begin{equation}
\label{eq:c4}
\lnot Qd
\end{equation}
\end{minipage}
\hfill
\begin{minipage}{.48\linewidth}
\begin{equation}
\label{eq:c5}
\lnot Pfc
\end{equation}
\begin{equation}
\label{eq:c6}
\forall x.~\lnot Rxc
\end{equation}
\begin{equation}
\label{eq:c7}
\forall x.~\lnot Rdx
\end{equation}
\end{minipage}
\\
\\
\noindent
Figure~\ref{fig:running} shows a connection tableau built from this set of clauses.
It has a single open branch $Rxy$, shown boxed at position~3. The current substitution is $\{~x \mapsto c,~y \mapsto d,~z \mapsto fc,~w \mapsto x~\}$. The only available extension steps for $Rxy$ are $\lnot Rdv$ or $\lnot Rvc$. The tableau is stuck, as neither are possible. At an earlier stage of construction, before $x \mapsto c$ and $y \mapsto d$, both extensions would have been possible.
Note that the way in which $Pz$ at position 4 is closed is irrelevant, while the sub-tableaux at 1 and 2 contribute to the dead end.

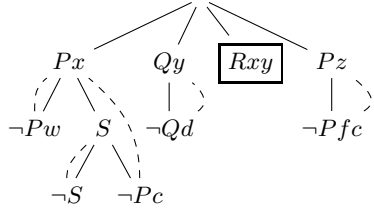
\begin{figure}
\centering
\begin{forest}
[,
	[$Px$, name=px
		[$\lnot Pw$, name=npw]
		[$S$, name=s
			[$\lnot S$, name=ns]
			[$\lnot Pc$, name=npc]
		]
	],
	[$Qy$, name=qy
		[$\lnot Qd$, name=nqd]
	],
	[\fbox{$Rxy$}, name=nrxy]
	[$Pz$, name=pz
		[$\lnot Pfc$, name=npfc]
	]
]
\draw[dashed] (npw) to[out=north, in=south west] (px);
\draw[dashed] (ns) to[out=north, in=south west] (s);
\draw[dashed] (npc) to[out=north, in=south east] (px);
\draw[dashed] (nqd) to[out=north east, in=south east] (qy);
\draw[dashed] (npfc) to[out=north east, in=south east] (pz);
\end{forest}
\caption{Running example: a connection tableau built from clauses \ref{eq:c1}--\ref{eq:c7}.}
\label{fig:running}
\end{figure}

\section{Learning Constraints}
\label{sec:constraints}
We propose learning and storing constraints during proof search in the connection tableau calculus in order to prevent us from repeatedly reaching dead ends for similar reasons.
We will now define a constraint language to explain why no inference step is possible in a given situation, which will allow us to design an improved search procedure in Section~\ref{sec:search-with-constraints}.
Suppose that a particular inference step would normally be applicable to an open branch in the tableau.
If this step is not applicable, it must be that some rule applications elsewhere in the tableau prevented it.
We therefore define our constraint language to be based on the inference rules of the connection tableau calculus.
We will refine this language in Section~\ref{sec:constraint-refinements}, but for now, consider the following definition.
\newcommand{\start}[1]{\mathcal{S}_{#1}}
\newcommand{\reduce}[2]{\mathcal{R}^{#2}_{#1}}
\newcommand{\extend}[3]{\mathcal{E}^{#1}_{#2/#3}}
\begin{definition}[Simplified Constraint Language]
\emph{Constraints} are sets of atoms. Each atom is either:
\begin{enumerate}
\item $\start{C}$, representing starting the tableau with clause $C$; 
\item $\reduce{p}{q}$, representing a reduction from position $p$ to an ancestor $q$ in the tableau;
\item $\extend{p}{C}{i}$, representing extending position $p$ by a connection to the $i$\textsuperscript{th} literal of clause $C$.
\end{enumerate}
\end{definition}
Note that each atom includes the open goal (or root) to which the step is applied.
This language is sufficient to explain why an inference $j$ that would be possible otherwise is currently not possible within the tableau and describes this situation in a way to cover a whole class of similarly affected tableaux.
This is done by finding a subset of the inference steps already applied to the tableau that prevent the application of inference $j$.
\begin{definition}[Reasons for failed inferences]
\label{def:inference-reason}
Take an open branch $B$ in a tableau $T$ constructed by a series of inferences $I$.
We construct a sub-tableau $T'$ by applying only those inferences $I' \subseteq I$ which are necessary to produce $B$, i.e. the start clause and a series of extensions along the path to $B$.
Suppose that there is an inference $j$ that can be applied to $B$ in $T'$ but not in $T$.
A \emph{reason} for failing to apply $j$ in $T$ is a minimal set $E \subseteq I \setminus I'$ which, if applied additionally to $T'$, prevents applying $j$ at $B$ in the resulting tableau.
\end{definition}
\begin{example}[Reasons for inference failure 1]
\label{ex:inference1}
Consider the tableau in Figure~\ref{fig:running}.
To close it, the remaining open branch $Rxy$ must be extended.
Suppose that we wish to extend it with clause \ref{eq:c7}.
This would have been initially possible, but by now the global substitution contains $x \mapsto c$ and the extension is impossible.

We can explain this in our language by noticing that the minimal set of previous inferences required to make the extension impossible are those that close the branch at position 1, $Px$.
The other two branches at 2 and 3 are irrelevant, even though those branches are also closed and affect the global substitution.
If we take the minimal set of previous inferences, we obtain $\{~ \start{\ref{eq:c1}}, \extend{1}{\ref{eq:c2}}{1}, \extend{1.2}{\ref{eq:c3}}{1}, \reduce{1.2.2}{1} ~\}$
\end{example}
\begin{example}[Reasons for inference failure 2]
\label{ex:inference2}
We return to the tableau in Figure~\ref{fig:running}, but now consider extending the open branch with the unit clause \ref{eq:c6}.
Again, we notice that only the branch $Qy$ at position 2 is relevant for explaining why this extension fails, and produce the reason $\{~ \start{\ref{eq:c1}}, \extend{2}{\ref{eq:c4}}{1} ~\}$.
\end{example}
Now we turn our attention to explaining why a tableau is stuck.
This has two parts: stating that there is an open branch $B$, and showing that no inference can be applied to $B$.
\begin{definition}[Reasons for stuck tableaux]
\label{def:tableau-reason}
Take $T$, $B$, $I$, $T'$ and $I'$ from Definition~\ref{def:inference-reason}, and take $J$ to be the set of possible inferences from $B$ in $T'$.
Suppose that no $j \in J$ can be applied to $B$ in $T$.
We compute the {set of reasons} $R_j$ for each inference $j$ as follows:
\begin{enumerate}
\item If the calculus \emph{prevents} $j$ in $T$, we compute an explanation $R_j$ by Definition~\ref{def:inference-reason}.
\item If $j$ is \emph{applicable but leads to a tableau that is also stuck}, we compute a reason $R'$ for that tableau recursively and set $R_j = R' \setminus \{~ j ~\}$.
\end{enumerate}
$I'$ describes $B$ as an open branch, and the union of all inference failure reasons $R_j$ shows that no inference can be applied to $B$.
We define
\[ \bigcup_j R_j \cup I' \]
to be a reason that $T$ is \emph{stuck}.
\end{definition}
\begin{example}[Reasons for stuck tableaux]
Consider once again the tableau in Figure~\ref{fig:running} and its open branch.
If the only possible extensions are with clauses \ref{eq:c6} and \ref{eq:c7}, the tableau is stuck as neither can be applied here.
Using the reason set from Examples \ref{ex:inference1} and \ref{ex:inference2} and noting that $I'$ is simply $S_{\ref{eq:c1}}$, we obtain
\[\{~ \start{\ref{eq:c1}},~ \extend{1}{\ref{eq:c2}}{1},~ \extend{1.2}{\ref{eq:c3}}{1},~ \reduce{1.2.2}{1},~ \extend{2}{\ref{eq:c4}}{1} ~\}\]
as a reason for $T$ being stuck.
\end{example}
In general, reasons are not unique for any given stuck tableau, both because there could be more than one open branch, and because there could be more than one reason for a failed inference.
Choosing one reason suffices, but some are likely to be stronger than others.

\section{Search with Learned Constraints}
\label{sec:search-with-constraints}
\begin{algorithm}[t]
\caption{An iterative search routine for finding closed tableaux\label{alg:procedure}.}
\begin{algorithmic}
\STATE $T \gets $ empty tableau
\STATE constraints $\gets \emptyset$
\STATE trail $\gets $ nil
\REPEAT
	\STATE success $\gets$ \FALSE
    \STATE $B$ $\gets$ \texttt{select\_open\_branch}($T$)
	\STATE learn $\gets$ \texttt{explain}($T$, $B$)
	\FORALL{possible inferences $j$ at $B$ in $T$}
		\IF{\NOT \texttt{apply}($T$, $j$)}
			\STATE learn $\gets \mathrm{learn} \cup \mathtt{compute\_reason}(T, j)$
		\ELSIF{there is a conflict clause $C \cup \{j\}$ violated by $j$ and the trail}
			\STATE learn $\gets \mathrm{learn} \cup C$
		\ELSE
			\STATE trail $\gets j :: $ trail
			\STATE success $\gets$ \TRUE
			\STATE \textbf{break}
		\ENDIF
	\ENDFOR{}
	\IF{\NOT success}
		\WHILE{learn is violated}
			\STATE $i \gets $ pop(trail)
			\STATE undo($i$, $T$)
		\ENDWHILE
		\STATE \texttt{record\_learned\_clause}(learn)
	\ENDIF
\UNTIL{$T$ is closed or learn is empty}
\end{algorithmic}
\end{algorithm}
\noindent
In the previous section, we defined the constraint language so that stuck tableaux can be adequately explained.
The search algorithm should now be redesigned to make use of these learned constraints.
We implement something similar to that found in CDCL SAT solvers, SMT solvers, or constraint satisfaction systems.
Alongside the current tableau, we maintain a \emph{trail} of atoms that are true for the current tableau.

The search algorithm repeatedly applies rules, gets stuck, learns a constraint, and backtracks. It terminates when the tableau is closed or the empty constraint is learned.
We maintain the invariant that no learned constraint is \emph{violated} by the current trail: a constraint is violated if all of its atoms are contained in the trail.
In case the tableau is empty, a start clause is chosen. Otherwise, at each iteration, an open branch of the tableau is selected for reduction or extension.
If any such inference is possible, it is applied to the tableau, and the corresponding atom describing the result of the inference is added to the trail.
On the other hand, an inference $j$ may be impossible either:
\begin{enumerate}
\item because the calculus does not permit it, or
\item because adding its corresponding atom would violate a learned clause.
\end{enumerate}
In the first case, a reason for the failed inference is computed as in Definition~\ref{def:inference-reason}.
Otherwise, the reason for its failure is the learned constraint it would violate, \emph{minus} the atom corresponding to $j$ (Definition~\ref{def:tableau-reason})
In this way, when all possible inferences at an open branch have failed, a constraint against the stuck tableau is learned such that all the atoms in the constraint are on the trail.
To restore the invariant, the system backtracks until at least one violated atom is no longer on the trail.

The overall procedure is shown in Algorithm~\ref{alg:procedure}: \texttt{compute\_reason} computes a reason in the sense of Definition~\ref{def:inference-reason}, and \texttt{explain} computes $I'$ for $B$ as in Definition~\ref{def:tableau-reason}.
It is considerably more complex than the usual procedures for finding closed connection tableaux, particularly those embedded in Prolog via the ``lean'' methodology~\cite{leancop}.
However, it is simpler in one aspect: there is no need to remember alternative inferences at backtracking points, which can be quite involved if not implemented in terms of Prolog's existing backtracking mechanism~\cite{bare-metal-tableaux}.

\subsection{Resource Bounds}
Connection systems typically search by iterative deepening on a particular metric, such as the length of the longest branch.
A small limit is set initially, and then a system will begin search, bounded by the current limit.
If no tableau exists at one iterative deepening level, the limit is increased and search tried again.
We follow this approach here: our search algorithm looks for a tableau bounded by some maximum branch length, and if one does not exist, it will eventually terminate by learning the empty clause.

Constraints learned at one iterative deepening level cannot be reused for the next, as our approach would become incomplete.
It is possible to alter the constraint language in order to express constraints that (do not) depend on the depth limit, but this is of limited practical purpose for at least two reasons.
Firstly, because the search space at the next iterative deepening level tends to be much larger than the exhausted previous level, reusing constraints independent of the depth limit from the previous level does not help much.
Second, in practice, there are few such constraints.

\subsection{Soundness, Termination and Completeness}
We show that Algorithm~\ref{alg:procedure} terminates at any fixed depth limit and use this to show completeness.
Soundness is trivial, as the routine searches within an existing sound calculus.
\begin{lemma}[Termination]
Fix a depth limit. Algorithm~\ref{alg:procedure} terminates.
\end{lemma}
\proof
First, note that at any depth limit and for any finite set of input clauses, there is a finite number of possible tableaux, all of which are of finite size.
Because all tableaux are of finite size, the routine will eventually either close the tableau, terminating immediately, or become stuck.
When stuck, the routine learns a constraint which, by construction, is violated by the current trail, and therefore prevents reaching at least this particular tableau again.
As constraints are never forgotten within the same depth limit, and there are a finite number of possible tableaux, termination is guaranteed as the solver eventually learns constraints eliminating all possible tableaux.
\qed
\begin{lemma}
\label{lemma:closed-not-violated}
Learned constraints are not violated by any closed tableau reachable in the proof calculus at a given depth limit.
\end{lemma}
\proof
By induction on the derivation of learned constraints.
Suppose a learned constraint $C$ is violated by a closed tableau $T^\star$.
By definition, $C$ is a subset of the rules required to construct $T^\star$.
Construct the intermediate tableau $T$ generated by $C$.
Take the next inference step $j$ from $T$ towards $T^\star$.
In Definition~\ref{def:tableau-reason}, $C$ is justified on the basis that all inferences, including $j$, from $T$ are impossible, either because the calculus prevents it (Definition~\ref{def:inference-reason}), or because another constraint $C'$ is violated.
By the induction hypothesis, applying $j$ does not violate any such $C'$, so $j$ must not be legal in the calculus, and we have a contradiction.
\qed

\begin{theorem}[Completeness]
If a closed tableau exists at a depth limit, it will be found by Algorithm~\ref{alg:procedure}.
\end{theorem}
\proof
By termination and Lemma~\ref{lemma:closed-not-violated}.
\qed

\section{Refining the Constraint Language}
\label{sec:constraint-refinements}
The simple constraint language introduced in Section~\ref{sec:constraints} is sufficiently expressive to block classes of similar tableaux, but is quite specific to a particular tableau and fails to block all the similar tableaux we might like.
It is also quite clunky to work with and would be difficult to compute inference failure reasons efficiently in practice: see Section~\ref{sec:computing-explanations}.

We therefore decompose each atom into multiple smaller atoms of two kinds: placing literals at positions in a tableau, and binding variables to terms.
However, \emph{mutatis mutandis} the search procedure remains the same, pushing one or more such atoms onto the trail for any one inference.
As well as being simpler to implement, constraints can be much stronger as they do not block only a particular derivation of a tableau, but any tableau having particular literals and variable bindings.

\begin{definition}[Refined Constraint Language]
A constraint remains a set of atoms.
However, each atom is either
\begin{enumerate}
\item $L@p$, a literal $L$ being placed at position $p$
\item $x \mapsto t$, a variable $x$ is bound to a term $t$ where $t$ itself may be a variable.
\end{enumerate}
\end{definition}
Each inference of the connection tableau calculus can be expressed as some combination of these.
Adding clauses in start and extension rules is done by placing their literals at the corresponding positions.
Connections of literals in extension and reduction rules are applied by computing the required variable bindings.
\begin{example}[Refined constraint learning]
Take the tableau in Figure~\ref{fig:running}.
We will explain why it is stuck in terms of the new constraint language.
Extending $Rxy$ with $\lnot Rvc$ is not possible because $y \mapsto d$, which is on the trail because $Qy$ was connected to $\lnot Qd$.
Similarly, extending $Rxy$ with $\lnot Rdv$ is not possible because $x \mapsto c$.
To finish the explanation, we have to say why $Rxy$ needs to be closed in the first place, but this is straightforward: $Rxy@3$.
The final explanation is therefore
\[ \{~ Rxy@3,~x \mapsto c,~y \mapsto d ~\}. \]
\end{example}

\subsection{No-Connection Atoms}
In the proposed language, explaining why an open branch cannot be \emph{reduced} can become overly specific.
\begin{example}[Explaining reduction failure]
Consider an open branch $\lnot Pc$ at the depth limit, with path literals $Px$, $Qc$, $Rcd$, $S$ and a substitution containing $x \mapsto d$.
Suppose the positions from root to leaf are $p_1 \ldots p_5$.
Clearly $\lnot Pc$ cannot be reduced.
In the case of $Px$, the constraint contains somewhat useful information: if $x$ were not bound to $d$, this branch could be reduced.
However, for all other path literals, the only useful information is that they cannot be connected with $\lnot Pc$, but this is not in the language, and we must learn
\[ \{~ \lnot Pc@p_5,~Px@p_1,x \mapsto d,~Qc@p_2,~Rcd@p_3,~S@p_4 ~\}. \]
\end{example}
This kind of situation occurs often in practice and needlessly specialises the learned constraint to a particular sequence of path literals.
To avoid this problem, a new kind of atom $p \not \sim q$ is introduced, representing that no connection can ever be made between the two positions $p$ and $q$, regardless of the substitution.
Whenever a literal is added to the tableau at position $q$, its path is checked to see, which literals at positions $p$ it could be reduced with, and where this is impossible, $p \not \sim q$ is added to the trail.
In the above example, the learned constraint would be
\[ \{~ \lnot Pc@p_5,~Px@p_1,~x \mapsto d,~p_2 \not \sim p_5,~p_3 \not \sim p_5,~p_4 \not \sim p_5 ~\}. \]
and more general, as it does not specify which literals are at $p_2$, $p_3$, or $p_4$.

\subsection{Disequations}
Classical refinements such as regularity and eliminating tautologies greatly improve the performance of connection systems~\cite{handbook}.
These are classically implemented by means of \emph{disequations}.
We can support this naturally in the constraint language by adding disequation atoms of the form $s \neq t$.
When a disequation is falsified, backtracking can be induced by giving the disequation and the variable bindings required to falsify it as a learned constraint.

\section{Implementation and Experimental Validation}
\label{sec:implementation}
We implemented a prototype system \hopcop\footnote{\url{https://github.com/MichaelRawson/hopcop}, commit \texttt{a4a0f66}} to experiment with constraint learning.
So far, the implementation is imperative --- although we suspect a lean Prolog implementation may be possible via \texttt{assert/1}~\cite{iso-prolog} --- and owes much to implementation techniques found in the Bare Metal Tableaux Prover~\cite{bare-metal-tableaux} and \meancop~\cite{meancop}.
\hopcop implements the clausal connection tableau calculus, without cuts~\cite{restricted-backtracking} and starting from clauses derived from the conjecture.
With the obvious exception of constraint learning, \hopcop's search routine resembles that of \meancop, if the \meancop flags \verb|--conj --nopaths| are set\footnote{starting with the annotated conjecture clauses (\texttt{--conj}) and preventing input clauses reordering (\texttt{--nopaths})}.
\meancop also implements a lemma mechanism~\cite{handbook} that \hopcop so far lacks.
In Sections~\ref{sec:constraint-management} and \ref{sec:computing-explanations} we highlight two aspects of the implementation that may be of interest to implementers of similar systems.

\subsection{Constraint Management and Detecting Conflicts}
\label{sec:constraint-management}
A large number of constraints are learned during search, millions with a long enough time limit.
However, this seems to be less dramatic than in other settings such as SAT solving, and we did not find any benefit from attempting to garbage-collect old constraints, so \hopcop retains \emph{all} constraints it learns.

It is still necessary to efficiently find conflicts among this large set when adding atoms to the trail.
This is done by a 1-watched-literal scheme~\cite{watched-literal}: there is no need for the 2-watched-literal scheme popular in SAT solving, as all atoms have the same polarity and unit propagation is of little use.
More than one conflict may be found when adding an atom to the trail: it is worth trying to choose conflicts that minimise the resulting learned constraint.
\hopcop greedily chooses the conflict that adds the fewest atoms to the constraint learned so far.

\subsection{Computing Explanations}
\label{sec:computing-explanations}
To explain why an inference that connects two literals is not possible, a subset of the current substitution must be computed that prevents their unification.
It is no doubt possible to construct a complex variable-tracking scheme to do this quickly, but for our application, the following procedure is acceptably fast.
\begin{enumerate}
\item Unify the two literals using a new scratch substitution $\tau$. As the inference was possible with an empty substitution (but not with $\sigma$), this must succeed.
\item Record the current state of $\tau$, say $\tau_0$.
\item For every binding $x \mapsto t$ in $\sigma$ we:
\begin{enumerate}
	\item Try to unify $x$ and $t$ in $\tau$. If this succeeds, continue the loop at step 3.
	\item Reset $\tau$ to $\tau_0$ and try to unify $x$ and $t$ again. If successful, go to step 2.
	\item Exit the loop on failure, retaining $x \mapsto t$ in $\tau$.
\end{enumerate}
\end{enumerate}
After this procedure, $\tau$ should contain the necessary subset of $\sigma$.
A similar routine can be used to determine why a disequation is falsified.

\subsection{Experiments}
\begin{table}
\caption{The number of extension steps tried in order to determine that there is no closed tableau of a certain depth on \texttt{PUZ005-1}. A proof exists at depth 8.}
\label{tab:backtracking}
\centering
\setlength{\tabcolsep}{1em}
\begin{tabular}{r r r r r r r r r}
\toprule
depth & 1 & 2 & 3 & 4 & 5 & 6 & 7\\
\midrule
\meancop & 1 & 4 & 24 & 108 & 535 & 9,963 & 6,445,008\\
\hopcop & 1 & 4 & 89 & 495 & 2,309 & 10,066 & 48,517\\
\bottomrule
\end{tabular}
\end{table}
\noindent
Having gone to some effort to reduce backtracking in theory, we wish to know whether this also helps in practice.
We first manually inspected the behaviour of both \meancop and \hopcop on problems taken from the \texttt{PUZ} domain of TPTP.
\meancop reports the number of successfully applied extension steps required to exhaust each iterative deepening level, so we instrumented \hopcop to do the same.
Table~\ref{tab:backtracking} shows the number of steps required for \texttt{PUZ005-1}\footnote{the first CNF problem of moderate difficulty in the \texttt{PUZ} domain}.
At lower iterative deepening levels, the result is mixed due to differences in search decisions and \meancop's lemma rule, but \hopcop typically extends a clear lead at higher levels.
\meancop maintains a much higher rate of inference: our implementation is not highly optimised, but we suspect that the overhead of maintaining the learned constraints would cause significant inferences-per-second overhead compared to \meancop even if it were.

\hopcop also ran head-to-head against \meancop on several popular first-order benchmark sets: FOF and CNF problems from TPTP version 9.0.0~\cite{tptp}, the MPTP challenge problems~\cite{mptp} in \emph{bushy} and \emph{chainy} variants, and the \emph{Miz40} ATP-minimised set~\cite{rlcop}, of which \emph{M2k} is a subset.
Both systems were given a time limit of 10 seconds per problem and \meancop was configured with \verb|--conj --nopaths| (to better match \hopcop, see above).
We also ran \meancop with the additional \verb|--cut| argument, which we call \meancopcut: this renders \meancop incomplete in exchange for significantly reduced backtracking.

We do not wish to claim anything about the relative strength of the systems, only that the data are consistent with the hypothesis that the reduction in backtracking achieved overcomes the overhead in terms of inferences-per-second speed.
Table~\ref{tab:results} shows the number of solved problems.
Readers may also be interested in the very thorough experimental data in Färber's discussion of various backtracking schemes~\cite{meancop}.

\begin{table}
\caption{Theorems proved in 10 seconds by \hopcop, \meancop, and \meancopcut on various benchmark sets.}
\label{tab:results}
\centering
\setlength{\tabcolsep}{1em}
\begin{tabular}{r r r r r r}
\toprule
& M2k & Miz40 & bushy & chainy & TPTP \\
\midrule
	\meancop & 795 & 7,592 & 480 & 157 & 3,578 \\
	\meancopcut & 878 & 9,748 & 562 & 337 & 3,283 \\
	\hopcop & 1,050 & 13,040 & 589 & 203 & 4,026 \\
\bottomrule
\end{tabular}
\end{table}

\section{Related Work}
The most directly related work is the various fixed restrictions on backtracking in connection tableau~\cite{restricted-backtracking,meancop}: these are by nature incomplete, but effective.
Older techniques such as \emph{failure caching}~\cite{failure-caching} also achieve a reduction in backtracking and remain complete, but with different mechanisms.
The Go\'eland tableau system exchanges substitution information~\cite{goeland} between concurrent branch explorations: this is likely to reduce backtracking, but again with a different mechanism.
Encodings of connection-driven search into SAT or SMT, such as ChewTPTP~\cite{chewtptp}, are complete and will also learn constraints during proof search, but behave very differently and are mostly encoded upfront.

We ourselves have tried various encodings~\cite{spanning} of connection methods into SAT and SMT via \emph{user propagation}~\cite{z3-up,ipasir-up}.
We found that SAT/SMT solvers, even with a lazy user-propagator encoding, are not a good match for this kind of proof search, as their internal heuristics have no knowledge of the current state of the tableau.
The solver will, for instance, very happily decide or propagate variables that encode some sub-tableau completely disconnected from the current state.
Refined encodings such as in Section~\ref{sec:constraint-refinements} improve the encoding's performance, but allow the solver to partially apply inferences: we are not sure of the performance merits of this, but it is highly confusing.

The first author has also investigated \emph{generating} SAT clauses from instances of clauses found in connection tableau during search~\cite{satcop}, as a kind of instance-based method~\cite{instance-methods}.
When the set of SAT clauses becomes unsatisfiable, it shows that the input clause set was also unsatisfiable.
While an extension of this instance-generation approach can be used to influence the connection tableau search, it is not the core of the method, unlike the constraint learning approach here.
The two are largely orthogonal and could be combined profitably.

Other work that is highly related but may be confusing is the concept of \emph{backjumping} in modal and other tableau~\cite{backjumping-modal}: the concept and its origins appears to be similar, but it is used to \emph{avoid} logical conflicts on a branch when looking for a model, rather than to avoid getting stuck when looking for a closed tableau.
The \textsf{MeTTeL}\textsuperscript{2} tableau prover generator~\cite{mettel2} has generic support for a similar kind of backjumping, which it calls \emph{conflict-directed backjumping}.

\section{Outlook}
We have integrated a constraint learning approach to guide search and reduce backtracking into a prototype first-order connection theorem prover \hopcop and observe that it reduces the search space significantly, which translates into practical performance.
A trade-off is memory use: constraints have to be kept somewhere. This was not excessive in our experience, but may prevent running \hopcop on your iPod\textsuperscript{\textregistered}~\cite{pocket-reasoner}.

It is likely that other non-confluent tableau calculi may benefit from such an approach.
We would also be interested in the intersection of this kind of learning with the \emph{machine} kind of learning: constraint learning could reduce the options available to a learned heuristic, while a good learned heuristic might rapidly learn useful constraints.

There are some areas for further improvement to the technique.
The most irritating are the explicit positions present in the constraint language, which limits the application of learned constraints.
Eliminating this would require detecting conflicts modulo structurally equivalent positions, which we suspect may be difficult to do efficiently.
A future \emph{lean} implementation may help with this, perhaps based on the recent realisation that \emph{backjumping is exception handling}~\cite{backjumping-exception}.

\begin{credits}
\subsubsection{\ackname}
We are grateful to Michael F\"arber in particular for his \meancop tool and stimulating discussions on this and related topics. This research was funded in whole or in part by the  ERC Consolidator Grant ARTIST 101002685, the ERC Proof of Concept Grant LEARN 101213411, the TU Wien Doctoral College SecInt, the FWF SpyCoDe Grant 10.55776/F85,  the WWTF grant ForSmart   10.47379/ICT22007, and the Amazon Research Award 2023 QuAT. 

\subsubsection{\discintname}
The authors have no competing interests to declare that are relevant to the content of this article.
\end{credits}
\bibliographystyle{splncs04}
\bibliography{TABLEAUX25}
\end{document}